\documentclass[a4paper,UKenglish,cleveref, autoref, thm-restate]{lipics-v2021}

\usepackage{amsmath}
\usepackage{amssymb}
\usepackage{amsthm}
\usepackage{graphicx}
\usepackage{algorithmic}
\usepackage{mathrsfs}
\usepackage{braket}
\usepackage{xcolor}
\definecolor{black}{rgb}{0,0.6,0}
\usepackage{cite}
\usepackage{doi}
\usepackage{tikz}
\usetikzlibrary{quantikz2}

\newcommand{\Or}{\mathcal{O}}
\newcommand{\RR}{\mathbb{R}}

\newcommand{\tr}{\mathrm{tr}}
\newcommand{\dd}{\mathrm{d}}
\newcommand{\ZZ}{\mathbb{Z}}

\DeclarePairedDelimiter{\abs}{\vert}{\vert}
\DeclarePairedDelimiter{\norm}{\Vert}{\Vert}

\DeclareMathOperator\erf{erf}

\title{Stochastic Error Cancellation in Analog Quantum Simulation} 

\author{Yiyi Cai}{Institute for Quantum Information and Matter, California Institute of Technology, Pasadena, CA, USA \and Department of Electrical Engineering, California Institute of Technology, Pasadena, CA, USA}{yiyi@caltech.edu}{[https://orcid.org/0009-0003-4092-200X]}{}

\author{Yu Tong}{Institute for Quantum Information and Matter, California Institute of Technology, Pasadena, CA, USA}{yutong9410@gmail.com}{https://orcid.org/0000-0002-7555-9373}{}

\author{John Preskill}{Institute for Quantum Information and Matter, California Institute of Technology, Pasadena, CA, USA \and AWS Center for Quantum Computing, Pasadena, CA, USA}{preskill@caltech.edu}{https://orcid.org/0000-0002-2421-4762}{}

\authorrunning{Y. Cai, Y. Tong, and J. Preskill} 

\Copyright{Yiyi Cai, Yu Tong, and John Preskill} 



\category{} 

\relatedversion{} 

\acknowledgements{We thank Minh Tran, Matthias Caro, Mehdi Soleimanifar, Hsin-Yuan Huang, Andreas Elben, ChunJun Cao, and Christopher Pattison for helpful discussions. 
We thank Manuel Endres for pointing us to Ref.~\cite{shaw2023benchmarking}. Y.C. acknowledges funding from Arthur R. Adams Summer Undergraduate Research Fellowship. 
Y.T. acknowledges funding from U.S. Department of Energy Office of Science, Office of Advanced Scientific Computing Research (DE-SC0020290), and from U.S. Department of Energy, Office of Science, Basic Energy Sciences, under Award Number DE-SC0019374. J.P. acknowledges support from the U.S. Department of Energy Office of Science, Office of Advanced Scientific Computing Research (DE-NA0003525, DE-SC0020290), the U.S. Department of Energy, Office of Science, National Quantum Information Science Research Centers, Quantum Systems Accelerator, and the National Science Foundation (PHY-1733907).
The Institute for Quantum Information and Matter is an NSF Physics Frontiers Center.}

\nolinenumbers 

\EventEditors{Fr\'{e}d\'{e}ric Magniez and Alex Bredariol Grilo}
\EventNoEds{2}
\EventLongTitle{19th Conference on the Theory of Quantum Computation, Communication and Cryptography (TQC 2024)}
\EventShortTitle{TQC 2024}
\EventAcronym{TQC}
\EventYear{2024}
\EventDate{September 9--13, 2024}
\EventLocation{Okinawa, Japan}
\EventLogo{}
\SeriesVolume{310}
\ArticleNo{2}

\begin{document}

\maketitle

\begin{abstract}
Analog quantum simulation is a promising path towards solving classically intractable problems in many-body physics on near-term quantum devices. 
However, the presence of noise limits the size of the system and the length of time that can be simulated. 
In our work, we consider an error model in which the actual Hamiltonian of the simulator differs from the target Hamiltonian we want to simulate by small local perturbations, which are assumed to be random and unbiased. 
We analyze the error accumulated in observables in this setting and show that, due to stochastic error cancellation, with high probability the error scales as the square root of the number of qubits instead of linearly.  
We explore the concentration phenomenon of this error as well as its implications for local observables in the thermodynamic limit. Moreover, we show that stochastic error cancellation also manifests in the fidelity between the target state at the end of time-evolution and the actual state we obtain in the presence of noise. This indicates that, to reach a certain fidelity, more noise can be tolerated than implied by the worst-case bound if the noise comes from many statistically independent sources.
\end{abstract}

\section{Introduction}
\label{sec:intro}
Quantum computers are expected to outperform classical computers at solving certain problems of interest in physics, chemistry, and materials science. Simulating the dynamics of many-body quantum systems is an especially hard problem for classical computers, making quantum dynamics a particularly promising arena for seeking quantum advantage. Eventually, scalable fault-tolerant quantum computers will be able to perform accurate simulations of quantum dynamics, but these robust large-scale quantum machines are not likely to be available for many years. Meanwhile, what are the prospects for reaching quantum advantage using near-term quantum simulators that are not error-corrected?

Circuit-based quantum algorithms for quantum simulation offer great flexibility and can be error-corrected, but with current quantum technology analog quantum simulators may offer substantial advantage in the system size and time that can be achieved in simulation \cite{preskill2018qcandbeyond, daley2022qsimulationadv, cirac2012goals}. Analog quantum processors have tunable Hamiltonians running on quantum platforms, but need not have universal local control to perform informative simulations \cite{cirac2012goals, franca2020limitations, depalma2023limitations}. However, because these devices are not error-corrected, it is especially important to understand how errors accumulate during analog simulations of quantum dynamics.

Recently Trivedi et al.~used the Lieb-Robinson bound to show that the errors in expectation values of local observables can be independent of system size for short time evolution \cite{trivedi2022quantum}. They used an error model in which the actual Hamiltonian realized in the device differs from the desired target Hamiltonian by small local perturbations. 
More precisely, they considered a geometrically local Hamiltonian on a $d$-dimensional lattice with $N$ sites (each occupied by a qubit), and assumed that the actual Hamiltonian $H'$ and the target Hamiltonian $H$ are related through
\begin{equation}
\label{eq:hamiltonian_relation_cirac_paper}
    H' = H + \delta \sum_{i=1}^M  V_i.
\end{equation}
Here each $V_i$ is a local term with $\norm{V_i} \leq 1$, $M=\Or(N)$ denotes the number of independent error terms, and $\delta$ is a small number characterizing the magnitude of the local perturbations.  
One of their main conclusions is that the error in the expectation value of a local observable at time $t$ is at most $\Or(t^{d+1}\delta)$, where $t^{d+1}$ is essentially the volume of the local observable's Lieb-Robinson past light cone, and is independent of the system size. For a general observable that is not necessarily local, or for $t$ large enough so that information has the time to reach every part of the system, the error is at most $\Or(Nt\delta)$ as expected from first-order perturbation theory.

This result can be seen as a worst-case bound, which applies even if the small local perturbations are chosen adversarially to produce the largest possible error. However, this worst-case choice is unlikely to occur in practice. For estimating the accumulated error that should be anticipated under realistic conditions, it is often beneficial to consider a probabilistic error model rather than an adversarial one.  To be concrete, we consider the error model
\begin{equation}
\label{eq:hamiltonian_relation_this_project}
    H' = H + \sum_{i=1}^M g_i V_i.
\end{equation}
where in contrast to \eqref{eq:hamiltonian_relation_cirac_paper}, we assume that the local perturbations are stochastic and statistically independent, e.g., each $g_i$ is an independent Gaussian random variable with mean $0$ and standard deviation $\delta$. Instead of the worst case, we may now consider the accumulation of error in the average case. 
That is, we envision sampling $H'$ from an ensemble of possible Hamiltonians that might be realized in the device, estimating the error that is typical for this ensemble.
In other scenarios, for example in the analysis of the Trotter error in digital quantum simulations \cite{chen2021concentration,zhao2022hamiltonian}, the average-case error is found to be much better than the worst case, and the same can be expected in this analog setting. 
As an example, Fig.~\ref{fig:error_cancellation_demo} shows the difference between the worst-case and average-case error accumulation for time evolution in a one-dimensional Heisenberg spin system perturbed by a site-dependent magnetic field.

\begin{figure}
    \centering
    \includegraphics[scale=0.63]{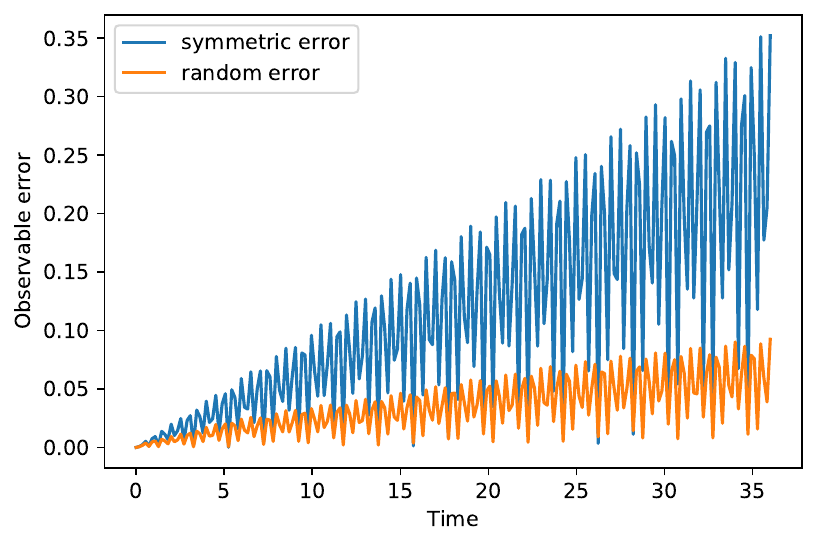}
    \caption{Numerical study of observable error of a 18-qubit system time-evolving under a Heisenberg XXX Hamiltonian $H=-J\sum_{i}\vec{\sigma}_i\cdot \vec{\sigma}_{i+1} + \sum_i (h+ g_i) X_i$. Local observable $Y_2$ is measured on the second qubit at time $t$,  with the system initial state being $\ket{0}^{\otimes N}$, and the local perturbations are $V_i=X_i$. The coupling strength $J$ between qubits is set to be $(0.2)\pi$, and $h=2\pi$. We compare the case of $g_i=\delta = 0.01$ (symmetric error), with a Gaussian model of random local perturbations, where $g_i \sim \mathcal{N}(0, \delta^{2})$ (random error). The rapid oscillations of the error as a function of time are discussed in Appendix \ref{app:oscillation}.}
    \label{fig:error_cancellation_demo}
\end{figure}

Simple classical reasoning provides an intuitive understanding of this finding. The cumulative effect of $M$ error sources, each contributing a Gaussian error with standard deviation $\delta$ and mean $0$, produces a total Gaussian error with mean $0$ and standard deviation $\sqrt{M}\delta$. For $M\gg 1$, this $\Or(\sqrt{M})$ cumulative error, resulting from stochastic error cancellations, is significantly suppressed compared to the $\Or(M)$ cumulative error which would occur in the absence of such cancellations. Because of the stochastic cancellations, we can tolerate more hardware error (larger $\delta$) than the worst-case error bound suggests.

In this paper, we explore the role of such error cancellations in analog quantum simulators and show that with high probability an error bound with square-root dependence on the system size $N$ can be achieved for general observables, in contrast to the linear $N$-dependence for the worst-case error bound. That is, the error bound is improved from $\Or(Nt\delta)$ to $\Or(\sqrt{N}t\delta)$. From this result, we derive an improved bound for local observables in the thermodynamic limit as well. Using the Lieb-Robinson bound, we show that the average-case error of local observables in the thermodynamic limit is bounded above by $\Or(t^{d/2+1}\delta)$ as opposed to the $\Or(t^{d+1}\delta)$ bound on the worst-case error. For fidelity, we show that the fidelity decays as $\exp(-\Or(\sqrt{N}\delta t))$ for small $\delta$ as $N$ increases, as observed in \cite{shaw2023benchmarking}, and is therefore slower than the exponential decay one would expect from the worst-case bound.

We are only aware of a few works besides \cite{trivedi2022quantum} that analyze the error in analog simulation. In \cite{poggi2020sensitivity} the error is analyzed by averaging over Haar random states, while in \cite{sarovar2017reliability} the leading-order error in a Gibbs state is expressed in terms of Fisher information. In contrast, we study how errors accumulate during time evolution of a quantum state.
The 
$\exp(-\Or(\sqrt{N}\delta t))$ decay of fidelity is observed in \cite{shaw2023benchmarking} both experimentally and numerically, though in a setting different from ours. In \cite{shaw2023benchmarking}, the dominant noise comes from the variation of the Rabi frequency, which is described by a single random variable. Using additional randomness introduced through the eigenstate thermalization hypothesis (ETH) \cite{DAlessioEtAl2016quantum} and neglecting oscillatory contributions, the authors are able to provide an explanation of the non-exponential decay. In our setting, we do not assume ETH or neglect any error term, but model the noise as coming from multiple statistically independent sources, and obtain a similar non-exponential decay of fidelity.

\section{Main results}
\label{sec:main_results}
Following \eqref{eq:hamiltonian_relation_this_project},
we consider a generalized setup of the local perturbation model, where each $g_i$ in \eqref{eq:hamiltonian_relation_this_project} is a $\chi$-deformed Gaussian defined as follows:

\begin{definition}[$\chi$-deformed Gaussian random variable]
    \label{defn:k-deformed_Gaussian}
    A random variable $g$ is a $\chi$-deformed Gaussian random variable if there exists $\theta \sim \mathcal{N}(0, 1)$ such that $g=\chi(\theta)$, and $\chi:\RR\to\RR$ is a strictly monotonic increasing differentiable function  satisfying 
    \begin{equation}
        \begin{aligned}
       & \abs{\dd \chi(\theta)/\dd \theta} \leq \delta, \\
       & \chi(0) = 0, \\
       & \chi(+\infty) = \Gamma,\ \chi(-\infty) = -\Gamma, \\
       & \mathbb{E}[\chi(\theta)] = 0.
        \end{aligned}
    \label{eq:k_requirements}
    \end{equation}
\end{definition}
Such a random variable $g$ has the nice properties that $|g|\leq \Gamma$ with probability $1$, and $|g|\leq \delta|\theta|$. We allow choosing $\Gamma=+\infty$.
This definition helps us generalize beyond the Gaussian noise model. Notably, we have the following examples that can be obtained as $\chi$-deformed Gaussian: (1) the uniform distribution $g_i \sim \mathcal{U}([-\delta', \delta'])$ where $\delta = \sqrt{2/\pi}\delta'$,  $\chi(\theta) = \delta'\erf(\theta/\sqrt{2})$, and $\Gamma=\delta'$; (2) the truncated Gaussian distribution for which $g_i$ obeys the Gaussian distribution $\mathcal{N}(0,\delta'^2)$ conditional on $|g_i|\leq \Gamma$, where $\delta = \delta'$ (one can in fact choose $\delta$ to be slightly smaller than $\delta'$), and $\chi(\theta)=\text{erf}^{-1} (\text{erf}(\frac{\theta}{\sqrt{2}})\text{erf}(\frac{\gamma}{\sqrt{2}\delta'})) \sqrt{2}\delta'$.

Denoting $H$ as a target Hamiltonian to be simulated and $H'$ as the actual Hamiltonian implemented, we show in Section \ref{sec:average_err} and \ref{sec:concentration} that 
\begin{theorem}
    \label{thm:bound_gen_obs}
    On a lattice consisting of $N$-sites, for Hamiltonians $H$ and $H'$ related through \eqref{eq:hamiltonian_relation_this_project} ($M=\Or(N)$), with each $g_i$ being an independent $\chi$-deformed Gaussian with $|\dd \chi(\theta)/\dd \theta|\leq\delta$,  $\Gamma\in(0,+\infty]$ (as defined in Definition~\ref{defn:k-deformed_Gaussian}), and
    $\sqrt{N}t\delta \leq \Or(1)$, we have
    \begin{equation}
        \big|\tr[\rho e^{iH't}Oe^{-iH't}] - \tr[\rho e^{iHt}Oe^{-iHt}]\big| \leq \Or(a\sqrt{N}t\delta\norm{O}) + \Or(Nt^2\delta^2\norm{O})
    \end{equation}
    with probability $1 - 2e^{-ca^2}$, for arbitrary $a>0$ and some absolute constant $c > 0$.
\end{theorem}
Note that Theorem~\ref{thm:bound_gen_obs} holds even when $\Gamma=\infty$. We assume $\sqrt{N}t\delta \leq \Or(1)$, which gives the time scale in which the simulation provides meaningful results. This result is stronger than the $\Or (Nt\delta \norm{O})$ scaling one would get without error cancellation, i.e. $g_{i} = \delta.$ This indicates that for given system size $N$ and time $t$, we can tolerate higher local perturbations up to $\frac{1}{\sqrt{N}t}$ instead of $\frac{1}{Nt}.$ 

Additionally, one may be interested in the thermodynamic limit ($N \rightarrow \infty$) as opposed to a finite system \cite{aharonov2022hamiltonian} and explore quantum simulation tasks that are stable against extensive errors. More precisely, for a local observable $O$ that is supported only on a constant number of sites, and a geometrically local Hamiltonian $H$, we want the error bound to be independent of the size of the system. With such an error bound, computing the expectation value of local observables in time evolution falls into the category of ``stable quantum simulation tasks'' as defined in \cite[Prop.~4]{trivedi2022quantum}.

An system-size independent error bound implies that the hardware error ($\delta$) does not need to be scaled down with system size, which is highly desirable for analog simulators. Specifically, we consider a bound for local observables acting on $\Or(1)$ adjacent sites in a quantum system on a lattice $\mathbb{Z}_{L}^{d}$, where $d$ is the lattice dimension and $L$ is the number of sites in each direction. Combining Theorem~\ref{thm:bound_gen_obs} with the Lieb-Robinson bound \cite{lieb1972velocity}, we show in Section~\ref{sec:local_observables} that the stability of the quantum task can be stated:

\begin{theorem}
    \label{thm:bound_local_obs}
    We consider a geometrically local Hamiltonian $H$ on a $d$-dimensional lattice $\ZZ_L^d$ with $L$ sites in each direction, and a local observable $O$ supported on $\Or(1)$ sites.
    The Hamiltonian can be written as $H=\sum_{\alpha\in \ZZ_L^d}H_{\alpha}$, where $\|H_{\alpha}\|=\Or(1)$ and $H_{\alpha}$ acts non-trivially only on sites that are within distance $r_0$ from $\alpha$, with $r_0=\Or(1)$.
    For a $H'$ related to $H$ through \eqref{eq:hamiltonian_relation_this_project} ($M=\Or(N)$), with each $g_i$ being an independent $\chi$-deformed Gaussian with $|\dd \chi(\theta)/\dd \theta|\leq\delta$,  $\Gamma=\Or(1)$ (as defined in Definition~\ref{defn:k-deformed_Gaussian}), $t^{d/2+1}\delta \leq \Or(1)$, and each site being acted on by only $\Or(1)$ of the error terms $V_i$,
    we have
    \begin{equation}
        \big|\tr[\rho e^{iH't}Oe^{-iH't}] - \tr[\rho e^{iHt}Oe^{-iHt}]\big| = \Or \left(at^{\frac{d}{2}+1}\delta\norm{O}\right) + \Or \left(at\delta \log^{d/2}(\delta^{-1})\norm{O}\right)
    \end{equation}
    with probability $1-2e^{-ca^{2}}$, for any $a>0$ and some absolute constant $c > 0$.
\end{theorem}

This is a stronger bound than the previously established one without error cancellation with leading term of $\Or \left(t^{d+1}\delta\right)$ \cite{trivedi2022quantum}. Note that in the above theorem we require that $\Gamma=\Or(1)$, as opposed to $\Gamma\in (0,\infty]$ in Theorem~\ref{thm:bound_gen_obs}. This is to ensure that the Lieb-Robinson bound can be applied to the Hamiltonian $H'$. $\Gamma=\Or(1)$ is physically justifiable because in realistic systems we do not expect to encounter an error that can be arbitrarily large.
For the fidelity decay, we have the following theorem:
\begin{theorem}
    \label{thm:fidelity_decay}
    On a lattice consisting of $N$-sites, for Hamiltonians $H$ and $H'$ related through \eqref{eq:hamiltonian_relation_this_project} ($M=\Or(N)$), with each $g_i$ being an independent $\chi$-deformed Gaussian with $|\dd \chi(\theta)/\dd \theta|\leq\delta$,  $\Gamma\in(0,+\infty]$ (as defined in Definition~\ref{defn:k-deformed_Gaussian}), and
    $a\sqrt{N}t\delta \leq \Delta$, where $\Delta$ is a constant that is independent of $a,N,t$, the fidelity
    \[
    F = |\braket{\phi(t)|\phi'(t)}|^2,
    \]
    where $\ket{\phi(t)}=e^{-iHt}\ket{\phi_0}$, $\ket{\phi'(t)}=e^{-iH't}\ket{\phi_0}$, for initial state $\ket{\phi_0}$,
    satisfies
    \begin{equation}
        F\geq e^{-{\Or}(a\sqrt{N} \delta t) - \Or(N\delta^{2}t^{2})}
    \end{equation}
    with probability $1 - 2e^{-ca^2}$, for arbitrary $a>0$ and some absolute constant $c > 0$.
\end{theorem}
We can see that up to leading order in $\delta$, the fidelity decays exponentially in $\sqrt{N}$ rather than $N$, thus showing a non-exponential decay of fidelity. From this we can see that in order to make the fidelity bounded away from $0$ by a constant, it suffices to have $\delta=\Or(1/(\sqrt{N}t))$, rather than $\Or(1/(Nt))$ that one would have with a worst-case bound. We will prove this theorem in Section~\ref{sec:non_exp_fidelity_decay}.

\section{The average error from random noise}
\label{sec:average_err}
We first consider the average observable error accumulated during time evolution and bound 
\begin{equation}
    \big|\mathbb{E}_{\{g_i\}}[\tr[\rho O'(t)]] - \tr[\rho O(t)] \big|
\end{equation}
with the notation
\begin{equation}
\label{eq:defn_time_evolved_operators_and_states}
    O(t) = e^{iH t} O e^{-iH t},\quad O'(t) = e^{i H't} O e^{-iH't}, \quad
    \rho(t) = e^{-iH t} \rho e^{iH t},\quad \rho'(t) = e^{-iH' t} \rho e^{iH' t}.
\end{equation}
We use the evolution under the target Hamiltonian $H$ as a reference frame, and consider the local perturbation in the interaction picture: 
\begin{equation}
    e^{-iH't} = e^{-iH t}\mathcal{T}e^{-i\int_0^t \sum_i g_i V_i(s)\dd s}
\end{equation}
where $V_i(s) = e^{iH s}V_i e^{-iH s}$ and $\mathcal{T}$ denotes time ordering.

We assume that $\delta\leq \Or(1/(\sqrt{N}t))$ in the analysis below. Because $M=\Or(N)$, we also have $\delta\leq \Or(1/(\sqrt{M}t))$.
We use the Dyson expansion to analyze the accumulation of error: 
\begin{equation}
    \begin{aligned}
    &\mathbb{E}[\tr[\rho O'(t)]] - \tr[\rho O(t)] \\
    &= \sum_{k=1}^{\infty} i^k \int_0^t\dd t_1\int_0^{t_1}\dd t_2\cdots \int_0^{t_{k-1}}\dd t_k \sum_{i_1,\cdots,i_k} \mathbb{E}[g_{i_1}\cdots g_{i_k}] \underbrace{\tr\big[\rho(t) [V_{i_1}(t_1),[\cdots [V_{i_k}(t_k),O]\cdots]\big]}_{C_{i_1i_2\cdots i_k}^{(k)}}.
    \end{aligned}
\end{equation}
 With $\|\rho(t)\|_{\tr}\leq 1$,\footnote{Here $\|\cdot\|_{\tr}$ denotes the trace norm.} 
Note that $\mathbb{E}[g_{i_1}\cdots g_{i_k}]$ is either $0$ (when $g_{i}$'s do not appear in pairs) or positive (when $g_i$'s appear in pairs), and therefore to upper bound the above quantity in absolute value we only need to upper bound $\big|C_{i_1i_2\cdots i_k}^{(k)}\big |$.
Because $||[A,B]|| \leq \norm{AB} + \norm{BA} \leq 2 ||AB||,$  
\begin{equation}
    \big|C_{i_1i_2\cdots i_k}^{(k)}\big |\leq ||\rho(t)||_{\tr}|| [V_{i_1}(t_1),[\cdots [V_{i_k}(t_k),O]\cdots]]|| \leq 2^{k} \norm{O}.
\end{equation}
Therefore
\begin{equation}
\label{eq:bounding_expect_err_with_moment_generating_function}
    \begin{aligned}
         \big|\mathbb{E}[\tr[\rho O'(t)]] - \tr[\rho O(t)] \big| &\leq \sum_{k=1}^{\infty} \frac{t^k}{k!} \mathbb{E}\left[\left(\sum_i g_{i}\right)^k\right] 2^k \|O\|\\
         &= \mathbb{E}[e^{2t\sum_{i=1}^{M} g_{i}}-1] \norm{O} \\
         &= \left(\prod_{i=1}^{M} \mathbb{E}[e^{2tg_{i}}]-1\right)\norm{O}.
    \end{aligned}
\end{equation}
Without loss of generality we assume that $\|O\|\leq 1$ hereafter.
From the above bound we can see that we only need to focus on bounding $\mathbb{E}[e^{2tg_{i}}]-1$ for each $i$.
Using the fact that $\mathbb{E}[g_i]=\mathbb{E}[\chi(\theta_{i})] = 0$ from \eqref{eq:k_requirements} and $|g_i| \leq \delta|\theta_{i}|$, we have
\begin{equation}
\label{eq:bounding_gi_exp_with_thetai}
    \begin{aligned}
        \mathbb{E}[e^{2tg_{i}}] &= \sum_{k=0}^{\infty} \frac{(2t)^{k}}{k!}\mathbb{E}[g_{i}^{k}] \leq 1+\sum_{k =2}^{\infty} \frac{(2\delta t)^{k}}{k!}\mathbb{E}[\abs{\theta_{i}}^{k}] = \mathbb{E}[e^{2\delta t\abs{\theta_{i}}}] - 2\delta t\mathbb{E}[\abs{\theta_{i}}].
    \end{aligned}
\end{equation}
Using Taylor's theorem in the Lagrange form, with the fact that
\begin{equation}
    \frac{\dd^k}{\dd a^k}\mathbb{E}[e^{a\abs{\theta_{i}}}] = \mathbb{E}[e^{a\abs{\theta_{i}}}\abs{\theta_{i}}^k],
\end{equation}
we have
\begin{equation}
    \mathbb{E}[e^{2\delta t\abs{\theta_{i}}}] - 2\delta t\mathbb{E}[\abs{\theta_{i}}] \leq 1+2\delta^2 t^2 \mathbb{E}[e^{2\delta t\abs{\theta_{i}}}\abs{\theta_{i}}^2].
\end{equation}
Because $\theta_i\sim\mathcal{N}(0,1)$, for any $a\geq 0$,
\begin{equation}
    \mathbb{E}[e^{a\abs{\theta_{i}}}\abs{\theta_{i}}^2] = (4a^2 e^{a^2})(1+\erf{a}) + 4\sqrt{\frac{2}{\pi}}ae^{-a^2/2} - \sqrt{\frac{2}{\pi}}a e^{a^2/2} = 1+\Or(a),
\end{equation}
we have (using $\delta\leq \Or(1/(\sqrt{M}t))$)
\begin{equation}
    \mathbb{E}[e^{2\delta t\abs{\theta_{i}}}] - 2\delta t\mathbb{E}[\abs{\theta_{i}}] \leq 1 + 2\delta^2 t^2 + \Or(\delta^3 t^3) = 1 + 2\delta^2 t^2 + \Or(M^{-3/2}).
\end{equation}
By 
\eqref{eq:bounding_expect_err_with_moment_generating_function}, \eqref{eq:bounding_gi_exp_with_thetai}, and $\|O\|\leq 1$ we then have
\begin{equation}
     \big|\mathbb{E}[\tr[\rho O'(t)]] - \tr[\rho O(t)] \big| \leq \left(1 + 2\delta^2 t^2 + \Or(M^{-3/2})\right)^M -1 = 2M\delta^2 t^2 + \Or(M^{-1/2}). 
\end{equation}

The above derivation leads us to the following theorem:
\begin{theorem}[Average error bound]

    On a lattice consisting of $N$-sites, for Hamiltonians $H$ and $H'$ related through \eqref{eq:hamiltonian_relation_this_project} ($M=\Or(N)$), with each $g_i$ being an independent $\chi$-deformed Gaussian with $|\dd \chi(\theta)/\dd \theta|\leq\delta$,  $\Gamma\in(0,+\infty]$, and
    $\sqrt{N}t\delta \leq \Or(1)$, we have
    \label{thm:avg_error_bound}
    \begin{equation}
        \big|\mathbb{E}_{\{g_i\}}[\tr[\rho O'(t)]] - \tr[\rho O(t)] \big| = \Or(N\delta^{2}t^{2} \norm{O})
    \end{equation}
\end{theorem}

This error bound shows that, if we average over multiple instances of the noise, then for the simulation to yield meaningful result up to time $t$ for a system with size $N$, we need local perturbation to be $\delta=\Or(1/(\sqrt{N}t))$, whereas the naive error bound of $\Or(Nt\delta)$ would only guarantee a meaningful result only when $\delta=\Or(1/(Nt))$. Therefore we can significantly extend the time and system size of the simulation that can be performed with guarantee at the same level of noise. 

\section{Concentration of the error}
\label{sec:concentration}
In the above section we focused on the expected error, but can the error be significantly larger than its expectation value? This is a question about the concentration of the probability measure, and our main tool is the following lemma: 

\begin{lemma}[Gaussian concentration inequality for Lipschitz functions]
\label{lem:concentration_inequality}
Let $f: \mathbb{R}^{M} \rightarrow \mathbb{R}$ be a function which is Lipschitz-continuous with constant $1$ (i.e. $|f(x)-f(y)| \leq |x-y|$ for all $x, y \in \mathbb{R}^{M}$), then for any $t$, 
\begin{equation}
    \mathbb{P}\left[|f(X) - \mathbb{E}[f(X)]| \geq t\right] \leq 2\exp(-ct^{2})
\end{equation}
for all $t > 0$ and some absolute constant $c > 0$, where $X \sim \mathcal{N}(0, 1)^{M}.$
\end{lemma}

The origin of this lemma is rather difficult to find, but its proof can be found at many places, including \cite[Theorem~2.1.12]{Tao2023TopicsRandomMatrixTheory} and \cite[Chapter~6, Theorem~2.1]{BandeiraLectureNotes}. It may appear at first glance that we might need a non-Gaussian version of this result, given that the noise we consider in Definition~\ref{defn:k-deformed_Gaussian} is not necessarily Gaussian. However, later we will see that a Gaussian version suffices because the noise can be regarded as a function of Gaussian random variables.

Recall that the expectation value $\tr[\rho O'(t)]$ is a function of the noise $\{g_i\}$, which is in turn a function of Gaussian random variables $\{\theta_i\}$ through $g_i = \chi(\theta_i)$. We therefore view $\tr[\rho O'(t)]$ as a function of $\{\theta_i\}$ which we denote by $h(\vec{\theta})$, where $\vec{\theta}=(\theta_1,\theta_2,\cdots,\theta_M)$. Similarly we denote $\vec{g}=(g_1,g_2,\cdots,g_M)$.
We will next proceed to obtain a Lipschitz constant for this function. Note that the Lipschitz constant can then be chosen to be the supremum of the 2-norm of the gradient, which we will justify below:
applying the mean value theorem (for several variables), for any pair of $\vec{\theta}$ and $\vec{\theta}'$ we have
\[
|h(\vec{\theta})-h(\vec{\theta}')| = |\nabla h(s\vec{\theta}+(1-s)\vec{\theta}') \cdot (\vec{\theta}-\vec{\theta}')|,
\]
where $\cdot$ denotes the Euclidean inner product (or the dot product). Therefore
\[
|h(\vec{\theta})-h(\vec{\theta}')| \leq \sup_{\vec{\theta}^*} |\nabla h(\vec{\theta}^*)| |\vec{\theta}-\vec{\theta}'|,
\]
where the norm $|\cdot |$ on the right-hand side denotes the vector 2-norm, and we have used the Cauchy-Schwarz inequality in arriving at this bound.
One can then choose the Lipschitz constant to be anything larger than or equal to $\sup_{\vec{\theta}^*} |\nabla h(\vec{\theta}^*)|$, i.e., any upper bound of $|\nabla h|$, which we will proceed to compute next. We will first bound individual partial derivatives

\begin{equation}
    \frac{\partial}{\partial\theta_i}h(\vec{\theta}) = \frac{\dd g_i}{\dd \theta_i}\partial_{g_i}\tr[Oe^{-iH'(\vec{g})t}\rho e^{iH'(\vec{g})t}],
\end{equation}
where we make explicit the $\vec{g}$-dependence in $H'$ defined in \eqref{eq:hamiltonian_relation_this_project}. Because $|\dd g_i/\dd \theta_i|\leq \delta$ by \eqref{eq:k_requirements}, we only need to bound $\partial_{g_i}\tr[Oe^{-iH'(\vec{g})t}\rho e^{iH'(\vec{g})t}]$:
\begin{equation}
\begin{aligned}
    |\partial_{g_{i}}\tr[O e^{iH'(\vec{g})t}\rho e^{-iH'(\vec{g})t}]| &\leq \norm{(\partial_{g_{i}}e^{iH'(\vec{g})t})O e^{-iH'(\vec{g})t}} + \norm{e^{iH'(\vec{g})t}O(\partial_{g_{i}} e^{-iH'(\vec{g})t})} \\
    &\leq 2\norm{O} \norm{\partial_{g_{i}} e^{-iH'(\vec{g})t}} 
    \leq 2\norm{O}t.
\end{aligned}
\end{equation}
In the last inequality above we used the fact that
\[
\|\partial_{g_{i}}e^{-iH'(\vec{g})t}\| = \|\int_{0}^t e^{-iH'(\vec{g})(t-s)}V_i e^{-iH'(\vec{g})s}\dd s\|\leq t,
\]
where we have used $\|V_i\|\le 1$.
We therefore have $|\partial_{\theta_i}h(\vec{\theta})|\leq 2\norm{O}t\delta$. As a result,
\begin{equation}
     |\nabla h|=\sqrt{\sum_{i=1}^{M} \abs{\partial_{\theta_{i}}h}^{2}} \leq 2\norm{O}\sqrt{M}t\delta.
\end{equation}
Because this holds for all choices of $\vec{\theta}$, the right-hand side is an upper bound of the supremum as well. Therefore we can choose the Lipschitz constant of $h$ to be $C_{\mathrm{Lip}}=2\norm{O}\sqrt{M}t\delta$.
$h(\vec{\theta})/C_{\mathrm{Lip}}$ then has Lipschitz constant $1$.
Through a direct application of Lemma \ref{lem:concentration_inequality}, we obtain that for some absolute constant $c > 0$ and any $a > 0$, 
\begin{equation}
\begin{aligned}
    \mathbb{P}[|h(\vec{\theta}) - \mathbb{E}[h(\vec{\theta})]| \geq a C_{\mathrm{Lip}}]  \leq 2e^{-ca^2}.
\end{aligned}
\end{equation}
We then have the following result, where we also use $M=\Or(N)$:
\begin{theorem}[Concentration bound of observable error]
\label{thm:concentration_bound}
    On a lattice consisting of $N$-sites, for Hamiltonians $H$ and $H'$ related through \eqref{eq:hamiltonian_relation_this_project} ($M=\Or(N)$), with each $g_i$ being an independent $\chi$-deformed Gaussian with $|\dd \chi(\theta)/\dd \theta|\leq\delta$,  $\Gamma\in(0,+\infty]$, and
    $\sqrt{N}t\delta \leq 1$, we have
    \begin{equation}
    \big|\tr[\rho O'(t)] - \mathbb{E}[\tr[\rho O'(t)]]\big| \leq 2a\norm{O}\sqrt{M}t\delta = \Or(a\norm{O}\sqrt{N}t\delta)
    \label{concentration_bound}
    \end{equation}
with probability $1 - 2e^{-ca^2}$, for arbitrary $a > 0$ and some absolute constant $c > 0$. 
\end{theorem}

Combining Theorem~\ref{thm:avg_error_bound} and Theorem~\ref{thm:concentration_bound}, we arrive at the result stated in Theorem~\ref{thm:bound_gen_obs}.

\section{Local observables}
\label{sec:local_observables}
In this section, we will take locality into consideration to obtain an error bound for local observables that is independent of the system size.
Such an error bound is needed to make the simulation meaningful in the thermodynamic limit.
We restrict ourselves to spin systems with spatial locality, i.e. systems with Hamiltonians defined on a $d$-dimensional lattice with $N$ sites in total and $L$ sites in each direction, written as 
\begin{equation}
    \label{eq:k-local_ham}
    H = \sum_{\alpha \in \mathbb{Z}_{L}^{d}} H_{\alpha}
\end{equation}
where $\norm{H_{\alpha}} \leq \zeta$ and $H_{\alpha}$ only acts on spins within a distance $r_0$ from $\alpha$, and $r_0=\Or(1)$. A key tool we are going to use is the Lieb-Robinson bound:

\begin{lemma}[Lieb-Robinson Bound, Refs.~\cite{hastings2004lieb, bravyi2006lieb, trivedi2022quantum}]
    \label{lem:lieb_robinson}
    For any local operator $O$ with support $S_{O}$, and for any $R > 0$, there exist positive constants $u, v$ that depend only on the lattice such that
    \begin{equation}
        \norm{O(t) - O_{R}(t)} \leq \norm{O} \abs{S_{O}} e^{-\mu R}(e^{v\zeta t} - 1) 
    \end{equation}
    where $O_{R}(t) = e^{iH_{R}t}Oe^{-iH_{R}t}$ with $H_{R} = H - \sum_{\alpha | d(S_{H^{\alpha}}, S_{O}) \geq R} H_{\alpha}$ being the restriction of the Hamiltonian to a region within distance $R$ of $S_{O}.$
\end{lemma}
We can then apply the Lieb-Robinson bound (Lemma \ref{lem:lieb_robinson}) to approximate the Heisenberg picture evolution of local observables with that corresponding to the Hamiltonian truncated within their light cones. Specifically, we consider the Heisenberg picture of observable $O$ under the truncated Hamiltonian $H_{R}$ and $H_{R}'$, denoted as:
\begin{equation}
    O_{R}(t) = e^{iH_{R}t}Oe^{-iH_{R}t}, \quad
    O_{R}'(t) = e^{iH_{R}'t}Oe^{-iH_{R}'t}
\end{equation}
where we denote $H_{R}$ as the truncated Hamiltonian acting non-trivially only on sites within distance $R$ from $S_O$, and $H_{R}'$ as the Hamiltonian obtained from $H'$ through the same procedure.
Assuming $\abs{S_{O}} \leq \Or(1)$, and with $\norm{H_{\alpha}}\leq \zeta$ and $e^{-\mu k} \geq 0$, we arrive at
\begin{equation}
\label{eq:O_diffs_bc_H_truncation}
    \|O'(t) - O_{R}'(t) \| \leq \Or (\norm{O}e^{-\mu R + v\zeta t}),\quad
    \|O(t) - O_{R}(t) \| \leq \Or (\norm{O}e^{-\mu R + v\zeta t})
\end{equation}

These bounds then allow us to upper bound the resulting errors in the expectation values through
\[
\big|\tr[\rho O'(t)] - \tr[\rho O'_{R}(t)]\big|\leq \|O'(t) - O'_{R}(t) \|,\quad \big|\tr[\rho O(t)] - \tr[\rho O_{R}(t)]\big|\leq \|O(t) - O_{R}(t) \|,
\]
which is true for any quantum state $\rho$.
This is a consequence of the duality between the Schatten $1$-norm (the trace distance) and $\infty$-norm (the spectral norm).

Note that the Lieb-Robinson bound only holds when the strength of local terms does not grow with the size of the system, and this is the reason why we choose $\Gamma=\Or(1)$ in Theorem~\ref{thm:bound_local_obs}, which ensures that all local terms in the Hamiltonian are bounded by a constant that is independent of the system size.

Since now $H_R$ and $H_R'$ acts non-trivially only on $\Or(R^{d})$, making this the effective system size. 
We also need to assume that the noise is spread evenly across the whole system, which can be rigorously stated as each site being acted on by only $\Or(1)$ of the error terms $V_i$. Therefore there are only $\Or(R^{d})$ terms $V_i$  that come into the difference between $H_R$ and $H_R'$.
Consequently we can apply Theorem~\ref{thm:bound_gen_obs} to get
\begin{equation}
\label{eq:O_diff_bc_noise}
    \begin{aligned}
        \big|\tr[\rho O_{R}'(t)]- \tr[\rho O_{R}(t)]\big| &\leq \Or(a\sqrt{R^{d}} t\delta \norm{O}) + \Or(R^{d}t^{2}\delta^{2}\norm{O})
    \end{aligned}
\end{equation}
with probability $1-2e^{-ca^{2}},$ for some absolute constant $c > 0$ and any $a>0.$ 
Combining the above bounds \eqref{eq:O_diffs_bc_H_truncation} and \eqref{eq:O_diff_bc_noise} together, we obtain 
\begin{equation}
    \big|\tr[\rho O'(t)]- \tr[\rho O(t)]\big| \leq \Or(aR^{d/2}t\delta \norm{O} + \norm{O}e^{-\mu R + v\zeta t})
\end{equation}
Note that if we choose $R = \frac{1}{\mu}(v\zeta t+\log (\delta^{-1}))$, then we get 
\begin{equation}
    \big|\tr[\rho O'(t)]- \tr[\rho O(t)]\big| \leq \Or\left(a\norm{O}\left(\frac{v\zeta}{\mu}t+\frac{1}{\mu}\log(\delta^{-1})\right)^{d/2}t\delta + \delta \norm{O}\right).  
\end{equation}
Using the fact that $(\alpha+\beta)^{d} \leq \Or(\alpha^{d}+\beta^{d})$ when $\alpha, \beta$ are constants, we arrive at 
\begin{equation}
    \big|\tr[\rho O'(t)] - \tr[\rho O(t)]\big| \leq \Or \left(at^{\frac{d}{2}+1}\delta\norm{O}\right) + \Or \left(a t\delta \log^{d/2}(\delta^{-1})\norm{O}\right) 
\end{equation}
with probability $1-2e^{-ca^{2}}.$ Theorem~\ref{thm:bound_local_obs} then follows.  

\section{Non-exponential fidelity decay}
\label{sec:non_exp_fidelity_decay}

The stochastic error cancellation can also be observed in the fidelity between the target state and the actual state we get at the end of time-evolution, leading to a surprising non-exponential decay of the fidelity for small $\delta$. This is similar to the non-exponential fidelity decay observed in \cite{shaw2023benchmarking}. 
We consider the fidelity metric as  
\begin{equation}
\label{eq:fidelity_defn}
    F = | \langle {\phi (t)}|{\phi'(t)} \rangle |^{2}
\end{equation}
where 
\begin{equation}
    \ket{\phi(t)} = e^{-iHt} \ket{\phi_{0}}, \quad \ket{\phi'(t)} = e^{-iH't} \ket{\phi_{0}}
\end{equation}
are pure states with $\ket{\phi(t)}$ denoting the time-evolved state of interest and $\ket{\phi'(t)}$ denoting the state under local perturbation. Here $\ket{\phi_{0}}$ is the initial state of the system. 
We will then prove Theorem~\ref{thm:fidelity_decay}.

\begin{proof}[Proof of Theorem~\ref{thm:fidelity_decay}]
    We are interested in upper-bounding 
\begin{equation}
    1-F = 1- | \langle {\phi (t)}|{\phi'(t)} \rangle |^{2}.
\end{equation}
We follow similar steps as previous proofs and first bound its expectation value:
\begin{equation}
    \begin{aligned}
        &\big| \mathbb{E}_{\{g_i\}}[ \langle {\phi (t)}|{\phi'(t)} \rangle] -1 \big| \\
        &= \big| \mathbb{E}[ \bra{\phi_{0}} e^{iHt} e^{-iH't} \ket{\phi_{0}}] - 1 \big| \\
        &= \big| \mathbb{E}[\bra{\phi_{0}} \mathcal{T}e^{-i\int_0^t \sum_i g_i V_i(s) \dd s} \ket{\phi_{0}}] - 1 \big|\\
        &= \Big| \sum_{k=1}^{\infty} (-i)^k \int_0^t\dd t_1\int_0^{t_1}\dd t_2\cdots \int_0^{t_{k-1}}\dd t_k \sum_{i_1,\cdots,i_k} \mathbb{E}[g_{i_1}\cdots g_{i_k}] \bra{\phi_{0}}V_{i_1}(t_1)V_{i_2}(t_2) \cdots V_{i_k} (t_k) \ket{\phi_0}\Big|
    \end{aligned}
\end{equation}
Here, $|\bra{\phi_{0}}V_{i_1}(t_1)V_{i_2}(t_2) \cdots V_{i_k} (t_k) \ket{\phi_0}| \leq 1$ since each local term $\|V_{i}\| \leq 1$. Note that 
\[
\sum_{i_1,\cdots,i_k} \mathbb{E} [g_{i_1}\cdots g_{i_k}] = \mathbb{E} \left[\sum_{i_1,\cdots,i_k} g_{i_1}\cdots g_{i_k} \right] = \mathbb{E} \left[\left(\sum_{i=1}^{M} g_{i}\right)^{k}\right] \geq 0.
\]
Therefore, 
\begin{equation}
\label{eq:expected_fidelity_bound}
\begin{aligned}
    \big| \mathbb{E}_{\{g_i\}}[ \langle {\phi (t)}|{\phi'(t)} \rangle] -1 \big|  &\leq \mathbb{E} \left[\sum_{k=1}^{\infty} \frac{t^{k}}{k!} \left(\sum_{i=1}^{M} g_{i}\right)^{k}\right] \\
    &= \mathbb{E} \left[e^{\sum_{i=1}^{M} tg_{i}}\right] - 1 \\
    &= \prod_{i=1}^{M} \mathbb{E}\left[ e^{tg_{i}}\right] - 1\\
    &\leq \prod_{i=1}^M(\mathbb{E}[e^{\delta t\abs{\theta_{i}}}] - \delta t\mathbb{E}[\abs{\theta_{i}}]) - 1 \\
    &\leq \left(1 + \delta^2 t^2 + \Or(M^{-3/2})\right)^M -1 \\
    &= \Or(M\delta^2 t^2) .
\end{aligned}
\end{equation}
We can now bound the concentration of the fidelity, i.e, by how much $\langle \phi (t)|\phi'(t)\rangle$ can deviate from its expectation value with large probability.
Following the previous setup, we treat $\langle \phi (t)|\phi'(t)\rangle$ as a function of $\{\theta_{i}\},$ denoted by $h(\Vec{\theta})$, where $\Vec{\theta} = (\theta_{1}, \theta_{2}, ... , \theta_{M}).$ We aim to find a Lipschitz constant for this function. We have 
\begin{equation}
    |\partial_{g_{i}} \bra{\phi_{0}} e^{iHt} e^{-iH'(g)t} \ket{\phi_{0}}|  \leq \left|\int_0^t\bra{\phi_{0}}e^{iHt}  e^{-iH'(g)(t-s)}(-i)V_ie^{-iH'(g)s}  \ket{\phi_{0}}\dd s \right| \leq t
\end{equation}
and $|\dd g_i/\dd \theta_i|\leq \delta$, giving us $|\partial_{\theta_i}h(\vec{\theta})|\leq t\delta.$ The Lipschitz constant can then be chosen as 
\begin{equation}
     \sqrt{\sum_{i=1}^{M} \abs{\partial_{\theta_{i}}h}^{2}} \leq \sqrt{M}t\delta = C_{\mathrm{Lip}}.
\end{equation}
$h(\vec{\theta})/C_{\mathrm{Lip}}$ then has Lipschitz constant $1$.
We then obtain the following bound from Lemma~\ref{lem:concentration_inequality}: 
\begin{equation}
\label{eq:fidelity_concentration_bound}
    \mathbb{P}[| \langle \phi (t)|\phi'(t)\rangle - \mathbb{E}[ \langle \phi (t)|\phi'(t)\rangle]| \geq a\sqrt{M}t\delta]  \leq 2e^{-ca^2}.
\end{equation}
for some absolute constant $c > 0$ and any $a > 0.$
Combining the two bounds \ref{eq:expected_fidelity_bound} and \eqref{eq:fidelity_concentration_bound} derived above, we arrive at 
\[| \langle \phi (t)|\phi'(t)\rangle - 1| \leq 
 a\sqrt{M}\delta t + \Or(N\delta^{2}t^{2})\] with probability $1-2e^{-ca^{2}}.$ This gives us 
 \begin{equation}
 \begin{aligned}
     |\langle \phi (t)|\phi'(t)\rangle|^{2} &\geq (1-a\sqrt{M}\delta t + \Or(N\delta^{2}t^{2} ))^{2} \\
     &= 1 - 2a\sqrt{M}\delta t + \Or(N\delta^{2}t^{2})\\
     &= e^{-2a\sqrt{M} \delta t - \Or(N\delta^{2}t^{2})}, 
 \end{aligned}
 \end{equation}
 for $a\sqrt{M}\delta t=\Or(1)$. Using $M=\Or(N)$, we prove the inequality in the statement of the theorem.
\end{proof}

\section{Conclusion}
\label{sec:conclusion}
In this work, we considered the observable error bounds for analog quantum simulation under random coherent noise coming from independent sources. 
We showed that such randomness leads to improved scaling in error bounds due to stochastic error cancellation. 
We studied general observables without locality constraints as well as local observables, finding in both cases that average-case error bounds scale more favorably than worst-case error bounds.
Such cancellation indicates a higher tolerance of noise for simulation tasks on near-term analog quantum simulators than suggested by the worst-case bound.

Although our result substantially improves the previous state-of-the-art error bounds, there are still many factors that are not taken into consideration in our analysis. For example, in many-body localized systems \cite{PhysRev.109.1492,PhysRevB.21.2366,PhysRevLett.95.206603,BASKO20061126,https://doi.org/10.1002/andp.201700169,doi:10.1146/annurev-conmatphys-031214-014726,RevModPhys.91.021001}, our error bound based on the Lieb-Robinson light cone will not be able to capture the slow propagation of information, thus leading to an over-estimation of the error. In general, a tight analysis of the error would require understanding how operators spread in the system, which is a highly non-trivial and system-specific problem \cite{chen2021operator, parker2019universal, schuster2022operator, schuster2022many}. Phenomena such as thermalization should also play an important role, because if a subsystem thermalizes then the error on local operators in the subsystem should no longer accumulate over time. 
Symmetry has also been shown to be helpful in reducing error in both analog and digital quantum simulations \cite{tran2021symmetry,rotello2022symmetry}, and so has randomness in the simulation algorithm and the initial state \cite{chen2021concentration,childs2019faster,an2021time}. Our results for geometrically local Hamiltonians should be generalizable to the situation with power-law decaying interactions \cite{tran2019locality, chen2019quantum, gong2014locality, luitz2019locality, tran2021lieb, tran2020hierarchy, chen2019finite}, where the Lieb-Robinson bound is still available when the decay is fast enough.
These observations indicate that we may still be able to obtain more accurate characterizations of error accumulation in practical analog simulators.

In this work we focused on quantum systems consisting of qubits or qudits, but many realistic quantum systems involve infinitely many local degrees of freedom and unbounded operators in the Hamiltonian, which makes analysis more difficult \cite{tong2022provably,kuwahara2022optimal}. We hope to tackle this problem in future works.

Furthermore, we note that an approximate ground-state projection operator can be written as a linear combination of time evolution operators (a fact which is instrumental in the proof of the exponential clustering theorem and 1D area law \cite{hastings2006spectral,hastings2007area,arad2013area}) and that approximate ground-state projectors may be used in algorithms for preparing the ground state \cite{ge2019faster,lin2020near,lin2022heisenberg}. We therefore expect our results to be useful for analyzing how errors in the Hamiltonian affect expectation values of observables in the ground state. We also hope to extend our result to thermal states using the techniques employed in \cite{trivedi2022quantum}.



\bibliography{bib2doi}

\appendix
\newpage

\section{Separation of oscillation and growth in the observable error}
\label{app:oscillation}

\begin{figure}
    \centering
    \includegraphics[width=0.473\textwidth]{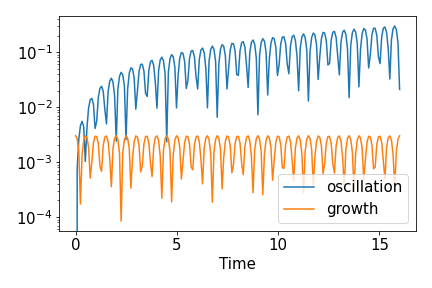}
    \includegraphics[width=0.456\textwidth]{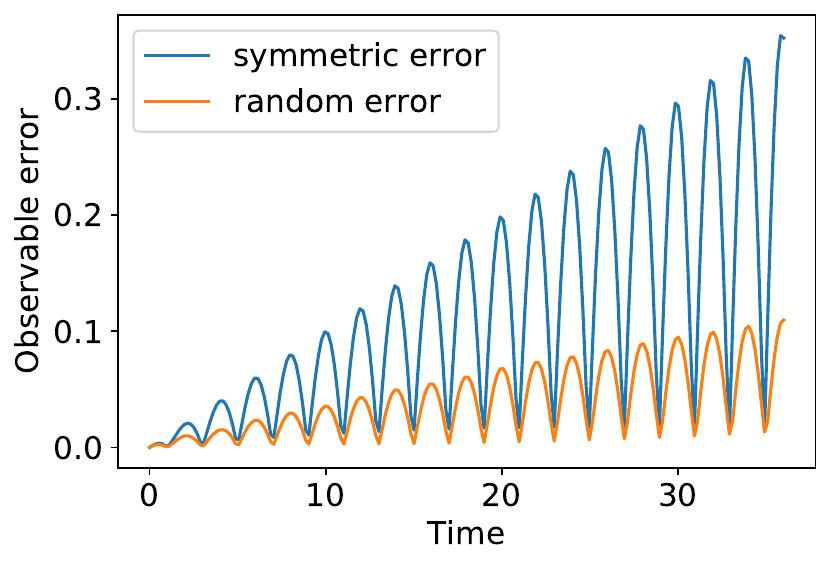}
    \caption{(Left) Comparing the oscillation part $|\braket{F(t)}|$ and the growth part $|\braket{G(t)}|$ in Eq.~\eqref{eq:decomp_oscillation_growth} that describes the evolution of the error operator. The setup is the same as in Figure~\ref{fig:error_cancellation_demo} except that the system now contains $8$ qubits and the observable is $(1/N)\sum_{i=1}^N Y_i$. (Right) Error in the observable expectation value for symmetric and random local errors. The simulation is performed with the same parameter setup as in Fig.~\ref{fig:error_cancellation_demo}, except with $h=(0.5)\pi$.}
    \label{fig:compare_oscillation_growth}
\end{figure}

In Fig.~\ref{fig:error_cancellation_demo} we observed that the error displays rapid oscillation in time. In this appendix we will investigate the cause of it.

We will examine how the operator $O$, the operator whose expectation value we want to estimate at the end of the evolution, evolves differently under the target Hamiltonian $H$ and the actual Hamiltonian $H'$. Using the notation introduced in Eq.~\eqref{eq:defn_time_evolved_operators_and_states}, we denote by $O(t)$ the time-evolved operator $O$ at time $t$ in the Heisenberg picture under the target Hamiltonian $H$, and by $O'(t)$ the corresponding operator under the actual Hamiltonian $H'$. We can write down an equation governing the error $O'(t)-O(t)$, from taking the time derivative in Eq.~\eqref{eq:defn_time_evolved_operators_and_states}:
\begin{equation}
\label{eq:decomp_oscillation_growth}
    \frac{\dd}{\dd t}(O'(t)-O(t)) = \underbrace{i[H,O'(t)-O(t)]}_{F(t)} + \underbrace{i\sum_{i=1}^M g_i[V_i,O'(t)]}_{G(t)}.
\end{equation}
We will show that only the second part $G(t)$ contributes to the growth of the error. Writing down the solution to the above differential equation using Duhammel's principle, for $0<s<t$ we have
\begin{equation}
    O'(t) - O(t) = e^{iH(t-s)}(O'(s)-O(s))e^{-iH(t-s)} + \int_{s}^t e^{iH(t-u)}G(u)e^{-iH(t-u)}\dd u.
\end{equation}
We observe that if $G(u)=0$ for $s<u<t$, then we would have $\|O'(t) - O(t)\|=\|O'(s)-O(s)\|$, and the error would not grow in magnitude. This shows that $G(t)$ is solely responsible for the growth of the error. The first term on the right-hand side of \eqref{eq:decomp_oscillation_growth} only rotates $O'(t) - O(t)$.

While $F(t)$ does not contribute to the growth of the error, it nevertheless plays a part in how the derivative changes, as can be seen from \eqref{eq:decomp_oscillation_growth}, which tells us that $\frac{\dd}{\dd t}\braket{O'(t)-O(t)} = \braket{F(t)}+\braket{G(t)}$. If $|\braket{F(t)}|\gg|\braket{G(t)}|$, then the error $\braket{O'(t)-O(t)}$ will be changing at a rate much faster than its growth, which indicates an oscillatory behavior. We numerically found that this is indeed the case. In Fig.~\ref{fig:compare_oscillation_growth}, we compare the magnitude of the oscillation part $|\braket{F(t)}|$ and the growth part $|\braket{G(t)}|$. We can see from the figure that $|\braket{F(t)}|\gg |\braket{G(t)}|$, which explains the rapid oscillation we see in Fig.~\ref{fig:error_cancellation_demo}. In particular, in the parameter setup of Fig.~\ref{fig:error_cancellation_demo}, we applied a large $X$-field whose strength is ten times the coupling constants. This $X$-field only contributes to $F(t)$ but not $G(t)$, which resulted in $|\braket{F(t)}|\gg |\braket{G(t)}|$. When we decrease the $X$-field strength the oscillation frequency decreases accordingly, as can be seen from the right panel of Fig.~\ref{fig:compare_oscillation_growth}.

\end{document}